\documentclass[journal,twoside,web]{ieeecolor}

\usepackage{lcsys}
\usepackage{cite}
\usepackage{amsmath,amssymb,amsfonts}
\usepackage{graphicx}
\usepackage{textcomp}

\usepackage{algorithm}
\usepackage{algpseudocode}

\usepackage{bm}

\usepackage{eucal}

\usepackage{graphicx}

\usepackage{mathtools}
\usepackage{newtxmath} 
\usepackage{multirow}

\usepackage{textcomp}
\usepackage{tikz}

\usepackage{pgfplots}
\pgfplotsset{compat=1.16}  

\usepackage[caption=false,font=footnotesize]{subfig}

\usepackage{hyperref}
\usepackage[acronym]{glossaries}
\usepackage[nameinlink]{cleveref} 

\creflabelformat{equation}{#2#1#3}
\crefname{equation}{Equation}{Equations}
\crefname{theorem}{Theorem}{Theorems}
\crefname{definition}{Definition}{Definitions}
\crefname{corollary}{Corollary}{Corollaries}
\crefname{lemma}{Lemma}{Lemmas}

\crefname{section}{Section}{Sections}

\crefname{table}{Table}{Tables}
\crefname{figure}{Figure}{Figures}
\crefname{algorithm}{Algorithm}{Algorithms}

\Crefname{table}{Table}{Tables}
\Crefname{figure}{Figure}{Figures}
\Crefname{algorithm}{Algorithm}{Algorithms}

\crefname{ineq}{Inequality}{Inequalities}
\Crefname{ineq}{Inequality}{Inequalities}
\creflabelformat{ineq}{#2{\upshape#1}#3}

\crefname{prob}{Problem}{Problems}
\Crefname{prob}{Problem}{Problems}
\creflabelformat{prob}{#2{\upshape#1}#3} 

\crefname{ass}{Assumption}{Assumptions}
\Crefname{ass}{Assumption}{Assumptions}
\creflabelformat{ass}{#2{\upshape#1}#3} 

\newtheorem{theorem}{Theorem}

\newtheorem{definition}{Definition}

\newtheorem{lemma}{Lemma}

\newtheorem{corollary}{Corollary}

\newtheorem{remark}{Remark}

\DeclarePairedDelimiter\set{\{}{\}}

\newcommand{\sthat}{\text{s.t.}}

\newcommand{\col}{\text{col}}

\newcommand{\diag}{\text{diag}}

\newcommand{\he}{\text{He}}

\newcommand{\norm}[1]{\left\| #1 \right\|}

\newcommand{\brackets}[1]{\left( #1 \right)}

\newcommand{\bracketl}[1]{\left[ #1 \right]}

\newcommand{\quadsupplyinput}[5]{\begin{bmatrix}
    #4 \\ #5
\end{bmatrix}^{\hspace{-1pt}T}\hspace{-4pt}\begin{bmatrix}
    #1 & #2 \\ * & #3
\end{bmatrix}\hspace{-3pt}\begin{bmatrix}
    #4 \\ #5
\end{bmatrix}}

\newcommand{\dbx}{\dot{\textbf{x}}}

\newcommand{\cH}{\mathcal{H}}

\newcommand{\cL}{\mathcal{L}}

\newcommand{\cT}{\mathcal{T}}

\newcommand{\cX}{\mathcal{X}}

\newcommand{\bbN}{\mathbb{N}}

\newcommand{\bbR}{\mathbb{R}}
\newcommand{\bbS}{\mathbb{S}}

\newcommand{\bzero}{\mathbf{0}}

\newcommand{\bA}{\mathbf{A}}
\newcommand{\bB}{\mathbf{B}}
\newcommand{\bC}{\mathbf{C}}
\newcommand{\bD}{\mathbf{D}}

\newcommand{\bF}{\mathbf{F}}
\newcommand{\bG}{\mathbf{G}}
\newcommand{\bH}{\mathbf{H}}
\newcommand{\bI}{\mathbf{I}}

\newcommand{\bK}{\mathbf{K}}
\newcommand{\bL}{\mathbf{L}}

\newcommand{\bN}{\mathbf{N}}

\newcommand{\bP}{\mathbf{P}}
\newcommand{\bQ}{\mathbf{Q}}
\newcommand{\bR}{\mathbf{R}}
\newcommand{\bS}{\mathbf{S}}

\newcommand{\bV}{\mathbf{V}}

\newcommand{\bX}{\mathbf{X}}
\newcommand{\bY}{\mathbf{Y}}
\newcommand{\bZ}{\mathbf{Z}}

\newcommand{\be}{\mathbf{e}}

\newcommand{\bn}{\mathbf{n}}

\newcommand{\bu}{\mathbf{u}}

\newcommand{\bx}{\mathbf{x}}
\newcommand{\by}{\mathbf{y}}
\newcommand{\bz}{\mathbf{z}}

\newcommand{\scrC}{\mathscr{C}}

\newcommand{\scrG}{\mathscr{G}}

\newcommand{\barbA}{\overline{\bA}}
\newcommand{\barbB}{\overline{\bB}}
\newcommand{\barbC}{\overline{\bC}}

\newcommand{\barbH}{\overline{\bH}}

\newcommand{\hatbe}{\widehat{\be}}

\newcommand{\hatbn}{\widehat{\bn}}

\newcommand{\hatbu}{\widehat{\bu}}

\newcommand{\hatbx}{\widehat{\bx}}
\newcommand{\hatby}{\widehat{\by}}
\newcommand{\hatbz}{\widehat{\bz}}

\newcommand{\hatbA}{\widehat{\bA}}
\newcommand{\hatbB}{\widehat{\bB}}
\newcommand{\hatbC}{\widehat{\bC}}
\newcommand{\hatbD}{\widehat{\bD}}

\newcommand{\hatbF}{\widehat{\bF}}

\newcommand{\hatbH}{\widehat{\bH}}

\newcommand{\hatbK}{\widehat{\bK}}

\newcommand{\hatbP}{\widehat{\bP}}
\newcommand{\hatbQ}{\widehat{\bQ}}
\newcommand{\hatbR}{\widehat{\bR}}
\newcommand{\hatbS}{\widehat{\bS}}

\newcommand{\tilbB}{\widetilde{\bB}}
\newcommand{\tilbC}{\widetilde{\bC}}

\newcommand{\tilbH}{\widetilde{\bH}}

\newcommand{\tilbK}{\widetilde{\bK}}

\newcommand{\thh}{\mathrm{th}}

\newacronym{ndt}{NDT}{Network Dissipativity Theorem}
\newacronym{io}{IO}{input-output}

\newacronym{lti}{LTI}{linear time-invariant}
\newacronym[plural=LMIs,longplural=linear matrix inequalities]{lmi}{LMI}{linear matrix inequality}
\newacronym{lqr}{LQR}{linear quadratic regulator}
\newacronym{lqg}{LQG}{Linear Quadratic Gaussian}
\newacronym{lq}{LQ}{linear quadratic}

\newacronym{admm}{ADMM}{alternating direction methods of multipliers}
\newacronym{sdp}{SDP}{semidefinite program}
\newacronym[plural=BMIs,longplural=bilinear matrix inequalities]{bmi}{BMI}{bilinear matrix inequality}
\newacronym{ico}{ICO}{iterative convex overbounding}

\newacronym{uav}{UAV}{unmanned aerial vehicles}
\newacronym{cps}{CPS}{cyber-physical system}
\usepackage{arydshln}
\crefformat{equation}{(#2#1#3)}

\newcommand{\haty}{\widehat{y}}

\newcommand{\dbX}{\delta\hspace{-.5pt}{\bX}}
\newcommand{\dbY}{\delta\hspace{-.5pt}{\bY}}

\newcommand{\dbF}{\delta\hspace{-.5pt}{\bF}}
\newcommand{\dbH}{\delta\hspace{-.5pt}{\bH}}
\newcommand{\dbQ}{\delta\hspace{-.5pt}{\bQ}}
\newcommand{\dbS}{\delta\hspace{-.5pt}{\bS}}
\newcommand{\dbR}{\delta\hspace{-.5pt}{\bR}}

\newcommand{\dPi}{\delta\hspace{-.5pt}{\Pi}}
\newcommand{\dtilbB}{\delta\hspace{-.5pt}{\tilbB}}
\newcommand{\dtilbC}{\delta\hspace{-.5pt}{\tilbC}}
\newcommand{\dhatbK}{\delta\hspace{-.5pt}{\hatbK}}
\newcommand{\dhatbF}{\delta\hspace{-.5pt}{\hatbF}}
\newcommand{\dhatbP}{\delta\hspace{-.5pt}{\hatbP}}
\newcommand{\dhatbQ}{\delta\hspace{-.5pt}{\hatbQ}}
\newcommand{\dhatbS}{\delta\hspace{-.5pt}{\hatbS}}
\newcommand{\dhatbR}{\delta\hspace{-.5pt}{\hatbR}}
\newcommand{\dhatbA}{\delta\hspace{-.5pt}{\hatbA}}
\newcommand{\dhatbB}{\delta\hspace{-.5pt}{\hatbB}}
\newcommand{\dhatbC}{\delta\hspace{-.5pt}{\hatbC}}
\newcommand{\dhatbD}{\delta\hspace{-.5pt}{\hatbD}}
\newcommand{\dtilbH}{\delta\hspace{-.5pt}{\tilbH}}
\newcommand{\dhatbH}{\delta\hspace{-.5pt}{\hatbH}}
\newcommand{\dbarbH}{\delta\hspace{-.5pt}{\barbH}}
\newcommand{\card}{\text{card}}

\def\BibTeX{{\rm B\kern-.05em{\sc i\kern-.025em b}\kern-.08em
    T\kern-.1667em\lower.7ex\hbox{E}\kern-.125emX}}
\begin{document}
\title{Communication-Aware Dissipative Output Feedback Control}
\author{Ingyu Jang, \IEEEmembership{Graduate Student Member, IEEE}, and Leila J. Bridgeman, \IEEEmembership{Member, IEEE}
\thanks{*This work is supported by ONR Grant No. N00014-23-1-2043.}
\thanks{Ingyu Jang (PhD Student), and Leila Bridgeman (Assistant Professor) are with the Department of Mechanical Engineering and Material Science, Duke University, Durham, NC 27708 USA
        (email: {\tt\small ij40@duke.edu; ljb48@duke.edu}, phone: 919-225-4215). }}

\maketitle

\begin{abstract}
Communication-aware control is essential to reduce costs and complexity in large-scale networks.  
This work proposes a method to design dissipativity-augmented output feedback controllers with reduced online communication.
The contributions of this work are three fold: a generalized well-posedness condition for the controller network, a convex relaxation for the constraints that infer stability of a network from dissapativity of its agents, and a synthesis algorithm integrating the Network Dissipativity Theorm, alternating direction method of multipliers, and iterative convex overbounding.
The proposed approach yields a sparsely interconnected controller that is both robust and applicable to networks with heterogeneous nonlinear agents.
The efficiency of these methods is demonstrated on heterogeneous networks with uncertain and unstable agents, and is compared to standard $\cH_\infty$ control.
\end{abstract}


\section{Introduction} \label{sec:introduction}
\IEEEPARstart{N}{etworked} control systems are increasingly prevalent in modern infrastructure, including smart grids, power plant networks, and swarm robotics. 
While centralized control schemes often suffer from scalability issues, such as prohibitive communication overhead, fully decentralized approaches can degrade closed-loop performance.
As networked systems grow in complexity, controller architectures must optimize the trade-off between performance and communication efficiency, specifically by promoting sparsity in agent interconnections \cite{jovanovic2016controller}. 
Many methods balance controller architecture against performance \cite{fardad2011sparsity,sadabadi2015fixed,babazadeh2016sparsity,lian2017game,lian2018sparsity,arastoo2016closed,locicero2020sparsity,lin2013design,eilbrecht2017distributed,negi2020sparsity,jang2025communication}, but most frameworks are restricted to \gls{lti} agents or employ full-state feedback, which is unrealistic in real physical situations.
This paper addresses the synthesis of robust, communication-aware controllers for networks with heterogeneous nonlinear agents under partial state feedback.

Imposing an $\ell_0$ norm (cardinality) constraint can directly yield the optimal sparse controller, but this results in NP-hard problems \cite{natarajan1995sparse}.
To circumvent this computational challenge, various sparsity-promoting methods have been developed, including $\ell_1$-norm relaxations \cite{fardad2011sparsity,sadabadi2015fixed,babazadeh2016sparsity}, gradient-based algorithms \cite{lian2017game,lian2018sparsity}, methods minimizing perturbations from a well-performing reference controller \cite{arastoo2016closed,locicero2020sparsity}, and \gls{admm} with sparsity penalties \cite{lin2013design,eilbrecht2017distributed,negi2020sparsity}.
However, the majority of existing works are restricted to \gls{lti} plants or assume full-state feedback frameworks.
Although some works \cite{eilbrecht2017distributed,locicero2020sparsity,sadabadi2015fixed} have extended these methods to observer-based control with sparse matrix parameters, the resulting architectures still exchange full observer states between agents. 
The more realistic setting is one where controller communication is restricted to local output information, similar to the plant interconnections, rather than full internal states.

Dissipativity \cite{willems1972dissipative,hill2003stability} offers a versatile framework for analyzing nonlinear dynamics by modeling systems as input–output operators rather than relying on internal state descriptions.
This perspective is especially powerful due to its compositional properties \cite{zakeri2022passivity}. 
The \gls{ndt} \cite{moylan1978stability,vidyasagar1981input} leverages this modularity to certify robust network stability using only open-loop dissipativity characteristics of individual agents.
By decoupling agent-level dynamics from the network topology, \gls{ndt} is uniquely suited for networked control with heterogeneous, nonlinear agents, as seen in  \cite{arcak2016networks,locicero2025dissipativity}, which applied \gls{ndt} to the design of centralized or decentralized controllers for large-scale networks.
It also offers a natural extension to distributed optimization paradigms \cite{boyd2011distributed}, which was leveraged in \cite{meissen2015compositional}.

Building upon \cite{jang2025communication}, which addressed dissipativity-based communication-aware control under full-state feedback, this paper considers the more realistic scenario of networked dynamics output-feedback control, in which controllers communicate only their filtered output information with others. 
Our objective is to identify the optimal \gls{io} communication links between local controllers while minimizing a global $\cH_\infty$-norm performance objective, by integrating \gls{ndt}, \gls{admm} \cite{boyd2011distributed}, and \gls{ico} \cite{warner2017iterative}.

The primary contributions of this work are threefold. 
First, we establish a generalized well-posedness condition for the controller network, ensuring a reasonable global control law as an extension of classical feedback well-posedness. 
Second, we derive a convex relaxation for the global $\cH_\infty$-norm constraint within a networked dynamics output-feedback framework, enabling efficient computation.
Third, we propose a computationally tractable synthesis algorithm that combines \gls{ndt}, \gls{admm}, and \gls{ico} to solve the sparse controller design problem.
This approach yields a sparsely interconnected controller that is robust and applicable to networks of nonlinear and heterogeneous agents.
Owing to the modularity of the \gls{ndt} framework, the proposed approach yields a sparsely interconnected controller that is both robust and applicable to networks with heterogeneous nonlinear agents. Furthermore, this modular structure ensures that the synthesis problem is readily extendable to distributed optimization paradigms.

\section{Preliminaries}
\subsection{Notation}
The sets of real and natural numbers up to $n$ are denoted by $\bbR$ and $\bbN_n$, respectively. 
The set of real $n{\times}m$ matrices is $\bbR^{n{\times} m}$. 
If $(\bA)_{ij}{\in}\bbR^{n_i{\times}m_j}$ and $\bA{\in}\bbR^{\sum_{i=1}^Nn_i{\times}\sum_{j=1}^Mm_j}$, then $(\bA)_{ij}$ is said to be a ``block'' of $\bA$, and $\bA$ is said to be in $\bbR^{N\times M}$ block-wise.
The set of $n{\times} n$ symmetric matrices is $\bbS^n$. 
The notation $\bA{\prec}0$ and $\bA{\preceq}0$ indicates that $\bA$ is negative definite and negative semi-definite, respectively. 
For brevity, $\he(\bA){=}\bA{+}\bA^T$ and asterisks, $*$, denote duplicate blocks in symmetric matrices.
$\cT_0^1(\bA)$ is the 1st order Taylor expansion of the matrix variable $\bA$ from its initial point $\bA^0$, meaning $\cT_0^1(\bA){=}\bA^0{+}\delta\bA$.
The set of square integrable functions is $\cL_{2}$. The Frobenius norm and $\cL_2$ norm are denoted by $\norm{\cdot}_F$ and $\norm{\cdot}_2$, respectively.  The truncation of a function $\textbf{y}(t)$ at $T$ is denoted by $\by_T(t)$, where $\by_T(t){=}\by(t)$ if $t{\leq}T$, and $\by_T(t){=}0$ otherwise. If $\|\by_T\|_2^2{=}{\langle}\by_T{,}\by_T{\rangle}{=}\int_0^{\infty}\by_T^T(t)\by_T(t)dt{<}\infty$ for all $T{\geq}0$, then $\by{\in}\cL_{2e}$, where $\cL_{2e}$ is the extended $\cL_2$ space.

\subsection{\texorpdfstring{$\bQ\bS\bR$}{QSR}-Dissipativity of Large-Scale Systems}
In this paper, controllers are synthesized based on the $\bQ\bS\bR$-dissipativity of each agent, defined as follows.
\begin{definition} [$QSR$-Dissipativity, \cite{vidyasagar1981input}] \label{def:dissipativity}
    Let $\bQ{\in}\bbS^l$, $\bR{\in}\bbS^m$, and $\bS{\in}\bbR^{l\times m}$. The operator $\scrG{:}\cL_{2e}^m{\mapsto}\cL_{2e}^l$ is $\bQ\bS\bR$-dissipative if there exists $\beta{\in}\bbR$ such that for all $\bu{\in}\cL_{2e}^m$ and $T{\geq}0$
    \begin{align} \label{eq:dissipativity}
        \int_0^T\quadsupplyinput{\bQ}{\bS}{\bR}{\scrG(\bu(t))}{\bu(t)}dt\geq\beta.
    \end{align}
\end{definition}

\cref{thm:stability_thm} relates dissipativity to \gls{io} stability, defined next.
\begin{definition} [\gls{io} or $\cL_2$-stability, \cite{lozano2013dissipative}] \label{def:io stability}
    An operator $\scrG:\cX_{e}^m{\mapsto}\cX_{e}^l$ is \gls{io}-stable, if for any $\bu{\in}\cX^m$ and all $\bx_0$ where $\cX$ is any semi-inner product space and $\cX_e$ is its extension, there exists a constant $\kappa{>}0$ and a function $\beta(\bx_0)$ such that
    \begin{align}
        \norm{(\scrG(\bu))_T}_\cX\leq\kappa\norm{\bu_T}_\cX+\beta(\bx_0)
    \end{align}
    where $\norm{\cdot}_\cX$ is the induced norm of the innerproduct space. If the space $\cX$ is $\cL_2$, then \gls{io} stability is called $\cL_2$ stability.
\end{definition}
\begin{theorem} \label{thm:stability_thm}
    The operator is $\cL_2$ stable if and only if it is $QSR$-dissipative with $\bQ\prec0$.
\end{theorem}

\gls{ndt}, stated next, shows how agent-level dissipativity extends to the network level, thereby guaranteeing the $\cL_2$-stability of the networked system.
\begin{theorem} [\gls{ndt}, \cite{moylan2003stability}] \label{thm:ndt}
    Consider $N$ $\bQ_i\bS_i\bR_i$ dissipative operators, $\scrG_i{:}\cL_{2e}^{m_i}{\mapsto}\cL_{2e}^{l_i}$, interconnected by matrices, $(\bH)_{ij}{:}\cL_{2e}^{l_j}{\mapsto}\cL_{2e}^{m_i}$ as
    \begin{align} \label{eq:interconnected_systems}
        \by_i{=}\scrG_i\bu_i,\quad\bu_i{=}\be_i{+}\hspace{-5pt}\sum_{j\in\bbN_N}\hspace{-3pt}(\bH)_{ij}\by_j,\quad\by{=}\scrG\be,\quad\bu{=}\be{+}\bH\by,
    \end{align}
    where $\bu{=}\col(\bu_i)_{i\in\bbN_N}$, $\by{=}\col(\by_i)_{i\in\bbN_N}$, $\be{=}\col(\be_i)_{i\in\bbN_N}$, and $\scrG{=}\diag(\scrG_i)_{i\in\bbN_N}$. Then, $\scrG{:}\cL_{2e}^m{\mapsto}\cL_{2e}^l$ is $\cL_2$ stable if 
    \begin{align} \label{eq:network_dissipativity_thm}
        \bQ+\bS\bH+\bH^T\bS^T+\bH^T\bR\bH\prec0
    \end{align}
    with $\bQ{=}\diag(\bQ_i)_{i\in\bbN_N}$, and $\bS$ and $\bR$ defined analogously.
\end{theorem}

\subsection{\texorpdfstring{\gls{ico}}{ico}} \label{subChap:convex_overbounding}
Optimal control synthesis problems frequently involve nonconvex \glspl{bmi} of the form
\begin{align} \label{eq:bmi}
    \bQ+\he(\bX\bN\bY)\prec0,
\end{align}
where $\bN{\in}\bbR^{p{\times}q}$ is fixed, and $\bQ{\in}\bbS^n$, $\bX{\in}\bbR^{n{\times}p}$. and $\bY{\in}\bbR^{q{\times}n}$ are design variables.
To handle the general NP-hardness of \cref{eq:bmi}, convex conservatishm can be introduced via \cref{thm:overbounding}.
\begin{theorem} [\cite{sebe2018sequential}] \label{thm:overbounding}
    Consider the matrices $\bQ{\in}\bbS^n$, $\bN{\in}\bbR^{p\times q}$, $\bX{\in}\bbR^{n\times p}$. and $\bY{\in}\bbR^{q\times n}$, where $\bQ$, $\bX$, and $\bY$ are design variables.
    The \gls{bmi} condition $\bQ{+}\he(\bX\bN\bY){\prec}0$ is implied by
    \begin{align} \label{eq:overbounding_sebe}
        \begin{bmatrix}
            \bQ & \bX\bN{+}\bY^T\bG^T \\
            \bN^T\bX^T{+}\bG\bY & -\he(\bG)
        \end{bmatrix}{\prec}0
    \end{align}
    for any $\bG{\in}\bbR^{q{\times}q}$ satisfying $\he(\bG){\succ}0$.
\end{theorem}

The conservative effect of \cref{eq:overbounding_sebe} can be mitigated by iteratively updating a feasible point, $(\bX^i,\bY^i)$, satisfying \cref{eq:bmi}.
With the update, $(\bX^{i+1},\bY^{i+1}){=}(\bX^i{+}\dbX,\bY^i{+}\dbY)$, $\dbX$ and $\dbY$ serve as the decision variables. The tightening of \cref{eq:overbounding_sebe} relative to \cref{eq:bmi} then lies in proportion to these perturbations. 
As detailed in \cite{warner2017iterative}, this iterative scheme reduces the conservatism inherent in \cref{thm:overbounding}.
Each optimization problem remains feasible since $\dbX{=}\bzero$ and $\dbY{=}\bzero$ yield the initial feasible point.

\begin{remark}
    In this paper, $\bI$ is used as $\bG$, but any $\bG$ satisfying $\he(\bG){\succ}0$ can be used for \cref{eq:overbounding_sebe}. 
\end{remark}


\section{Sparsity-Promoting Dissipativity-Augmented Control}
Consider a multi-agent networked system $\scrG$ consisting of $N$ heterogeneous agents $\scrG_i$ interconnected through $\bH$. The overall network dynamics are described by
\begin{align} \label{eq:plant}
    {\setlength{\arraycolsep}{4pt}\begin{array}{rlll}
        \scrG_i: & \dbx_i=f_i(\bx_i,\bu_i), & \by_i=h_i(\bx_i) \\
        \scrG: & \dbx=f(\bx,\be),&\bz=h(\bx), & \bu=\be+\bH\by,\;\;\bz=\by
    \end{array}}
\end{align}
where $\bx_i{\in}\cL_{2e}^{n_i}$, $\bu_i{\in}\cL_{2e}^{m_i}$, and $\by_i{\in}\cL_{2e}^{l_i}$ are the states, inputs, and outputs of the $i^\thh$ agent, respectively. $\bu$, $\be$, and $\by$ are stacked vectors defined in \cref{thm:ndt}, where $\be$ represents the exogenous input to the global network. The global network output is denoted by $\bz$, which is equivalent to $\by$.

\subsection{Well-Posedness of Controller Interconnection}
Consider a network of agents and local dynamic output feedback controllers, $\scrC_i$, that communicate through their output measurements. 
This paper aims to jointly design local controllers for each agent and a sparse network topology, so that the resulting global controller $\scrC$ stabilizes and regulates the network $\scrG$. 
\cref{lem:well_posed} provides conditions for well-posedness, which ensures the existence and uniqueness of the closed-loop solution. 

\begin{definition} [Chapter 5.2 \cite{khalil1996robust}]
    An interconnected system is said to be well-posed if all interconnected transfer matrices are well-defined and proper.
\end{definition}

\begin{lemma} \label{lem:well_posed}
    Consider $N$ \gls{lti} systems with minimal state-space realizations $\scrC_i{:}\dot\hatbx_i{=}\hatbA_i\hatbx_i{+}\hatbB_i\hatbu_i$, $\hatby_i{=}\hatbC_i\hatbx_i{+}\hatbD_i\hatbu_i$. 
    Construct the global \gls{lti} system $\scrC$ through the interconnections $\hatbu{=}\tilbH_y\hatbe{+}\underline\bH\hatby$ and $\hatbz{=}\tilbH_{\haty}\hatby$, where $\hatbe$ and $\hatbz$ are the input and output of $\scrC$, respectively, $\hatbu{=}\col(\hatbu_i)_{i\in\bbN_N}$, and $\hatby$, $\hatbx$, $\hatbe$, $\hatbz$ are defined analogously. Then, $\scrC$ is well-posed if and only if $\bI{-}\hatbD_d\underline\bH$ is invertible, where $\hatbD_d{=}\diag(\hatbD_i)_{i\in\bbN_N}$.
\end{lemma}
\begin{proof}
    $\scrC$ can be modeled as a closed-loop system with an external input, with gains $\tilbH_y$ and $\tilbH_{\haty}$ on the input and output, respectively.
    The plant is a \gls{lti} system with $(\hatbA_d{,}\hatbB_d{,}\hatbC_d{,}\hatbD_d)$, where $\hatbA_d$, $\hatbB_d$ and $\hatbC_d$ are defined analogously to $\hatbD_d$, and the feedback gain is $\underline\bH$. Then, the closed-loop realization of $\scrC$ with input $\hatbe$ and output $\hatbz$ is given by
    \begin{align*}
        {\setlength{\arraycolsep}{0.5pt}\begin{array}{ll}
            \hatbA{=}\hatbA_d{+}\hatbB_d\underline\bH(\bI{-}\hatbD_d\underline\bH)^{{-}1}\hatbC_d, &
            \hatbB{=}(\hatbB_d{+}\hatbB_d\underline\bH(\bI{-}\hatbD_d\underline\bH)^{{-}1}\hatbD_d)\tilbH_y, \\
            \hatbC{=}\tilbH_{\haty}(\bI{-}\hatbD_d\underline\bH)^{-1}\hatbC_d, &
            \hatbD{=}\tilbH_{\haty}(\bI{-}\hatbD_d\underline\bH)^{-1}\hatbD_d\tilbH_y. 
        \end{array}}   
    \end{align*}
    Since the plant has a state-space realization, the well-posedness of $\scrC$ is equivalent to the existence of $(\bI{-}\hatbD_d\underline\bH)^{-1}$ \cite[Chapter 5.2]{khalil1996robust}.
\end{proof}

To ensure well-posedness, each local controller $\scrC_i$ and the resulting global controller $\scrC$ are given by
\begin{align} \label{eq:controller}
    {\setlength{\arraycolsep}{2.5pt}\begin{array}{rllll}
        \scrC_i: & \dot\hatbx_i{=}\hatbA_i\hatbx_i{+}\hatbB_i\hatbu_i, & \hatby_i{=}\hatbC_i\hatbx_i{+}\hatbD_i\hatbu_i, \\
        \scrC: & \dot\hatbx{=}\hatbA\hatbx{+}\hatbB\hatbe,&\hatbz{=}\hatbC\hatbx{+}\hatbD\hatbe, & \hatbu{=}\tilbH_y\hatbe, & \hatbz{=}\tilbH_{\haty}\hatby,
    \end{array}}
\end{align}
where $\hatbx_i{\in}\cL_{2e}^{n_i}$, $\hatbu_i{\in}\cL_{2e}^{l_i}$, $\hatby_i{\in}\cL_{2e}^{m_i}$, $\hatbA$, $\hatbB$, $\hatbC$, and $\hatbD$ follows the setup in \cref{lem:well_posed}. $\scrC$ is always well-posed since $\underline\bH{=}\bzero$.

\subsection{Closed-loop Networked System}
Consider the multi-agent networked system $\scrG{:}\be{\to}\bz$, defined in \cref{eq:plant}, interconnected in feedback with the networked controller $\scrC{:}\hatbe{\to}\hatbz$, defined in \cref{eq:controller}. 
Let $\bn$ and $\hatbn$ be exogenous disturbances affecting the plant and controller inputs, respectively, so that $\be{=}\bn{+}\hatbz$ and $\hatbe{=}\hatbn{+}\bz$, as illustrated in \cref{fig1:closed_loop_network}.
This configuration can be equivalently viewed as a network interconnection among agents, $\scrG_i$, and their respective local controllers, $\scrC_i$.
The resulting global interconnection is 
\begin{align} \label{eq:interconnection}
    {\setlength{\arraycolsep}{3.5pt}
    \begin{bmatrix}
        \bu \\ \hatbu
    \end{bmatrix}
    {=}\begin{bmatrix}
        \bn{+}\bH\by{+}\tilbH_{\haty}\hatby \\
        \tilbH_y(\bz{+}\hatbn)
    \end{bmatrix}
    {=}\begin{bmatrix}
        \bn \\ \tilbH_y\hatbn
    \end{bmatrix}
    {+}\barbH
    \begin{bmatrix}
        \by \\ \hatby
    \end{bmatrix},\;
    \barbH{=}\begin{bmatrix}
        \bH & \tilbH_{\haty} \\
        \tilbH_{y} & \bzero
    \end{bmatrix},}
\end{align}
where we used $\bz{=}\by$.
In the global interconnection matrix $\barbH$, 
$(\bH)_{ii}{=}\bzero$ to preclude local self-feedback. 
This closed-loop architecture is illustrated in \cref{fig1:decomposed_network}.
The objective is to synthesize sparse structures for $\tilbH_y$ and $\tilbH_{\haty}$ within $\barbH$ so that the controller $\scrC$ stabilizes the plant $\scrG$ in the presence of $\bn$ and $\hatbn$, while improving the network performance.

\begin{figure}
    \captionsetup[sub]{aboveskip=2pt, belowskip=0pt,width=0.25\textwidth}
    \centering
    \subfloat[Closed-loop representation]{
    \includegraphics[width=0.23\textwidth]{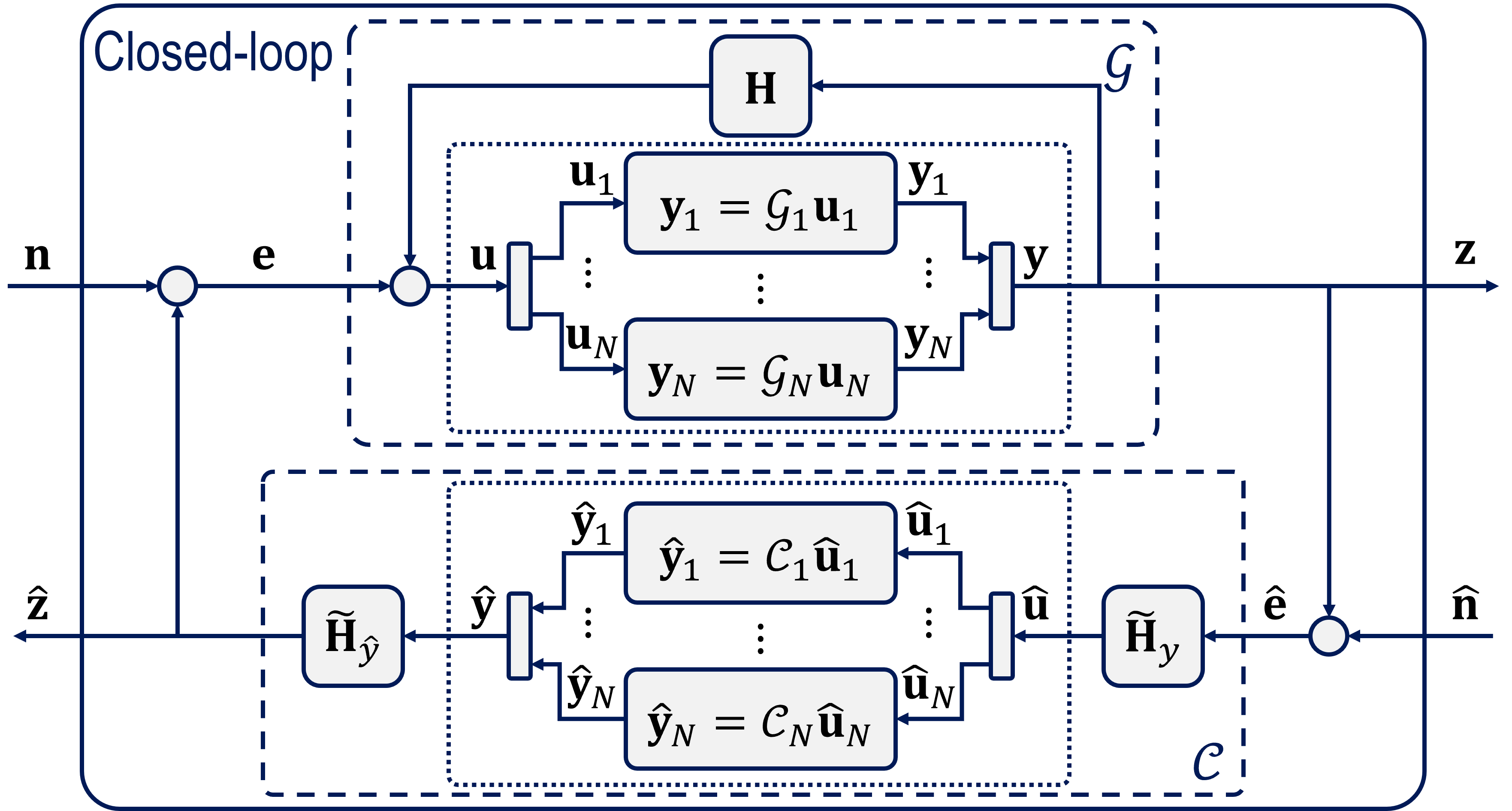} \label{fig1:closed_loop_network}
    }
    \subfloat[Decomposed representation]{
    \includegraphics[width=0.23\textwidth]{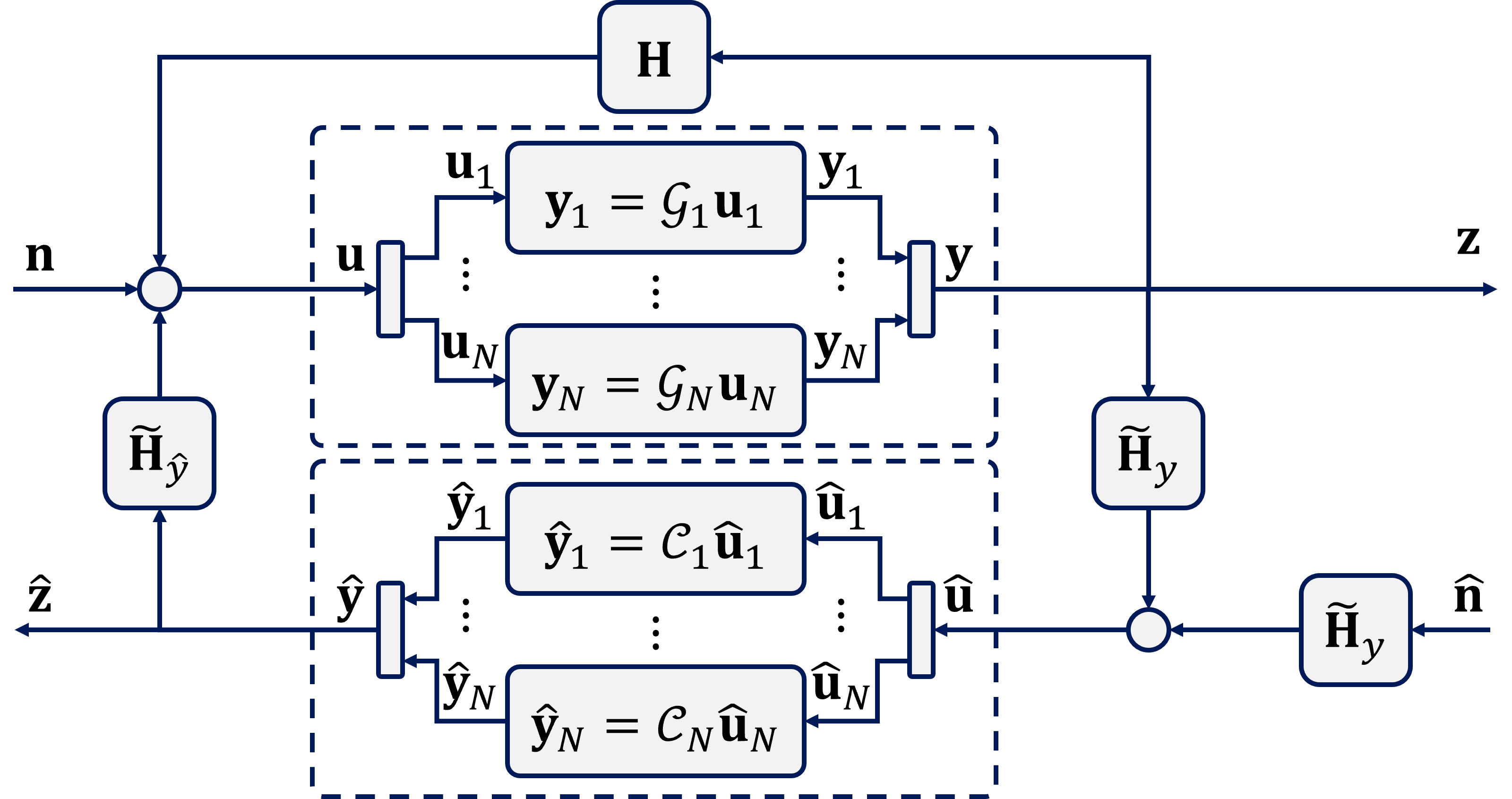} \label{fig1:decomposed_network}
    }
    \vspace{-2pt}
    \caption{Two representations of a multi-agent system and its controller.}
    \label{fig1:network}
    \vspace*{-1.25\baselineskip} 
\end{figure}

\subsection{Dual-Model Synthesis} \label{subchap:dual_model}
\gls{io} stability is closely linked to the $\bQ\bS\bR$-dissipativity properties of the agents and controllers.
If each agent is $\bQ_i\bS_i\bR_i$-dissipative, stability of the system can be guaranteed by enforcing $\hatbQ_i\hatbS_i\hatbR_i$-dissipativity on the corresponding local controllers and applying \gls{ndt} with $\barbH$.
However, according to \cref{eq:interconnection}, the external input to the system is $(\bn,\tilbH_y\hatbn)$, not $(\bn,\hatbn)$, which differs from the standard configuration of \gls{ndt}.
Nevertheless, \gls{ndt} remains applicable, because premultiplication by a time-invariant matrix preserves $\cL_2$-integrability of signals.

In this work, the objective of the controller is to attenuate the impact of $\bn$ and $\hatbn$ on the plant while guaranteeing \gls{io} stability of the network via \gls{ndt}, and this attenuation can be performed using a nominal linearized plant, $\scrG^{lti}$, described by
\begin{align} \label{eq:linear_plant}
    {\setlength{\arraycolsep}{.5pt}\begin{array}{rllll}
        \scrG_i^{lti}{:} & \dbx_i{=}\bA_i\bx_i{+}\bB_i\bu_i, & \by_i{=}\bC_i\bx_i, \\
        \scrG^{lti}{:} & \dbx{=}(\bA_d{+}\bB_d\bH\bC_d)\bx{+}\bB_d\be,  &\;\bz{=}\bC_d\bx, & \bu{=}\be{+}\bH\by,\,\bz{=}\by,
    \end{array}}
\end{align}
and a global controller $\scrC$ in \cref{eq:controller}. The parameters of the resulting closed loop of $\scrG^{lti}$ and $\scrC$ are
\begin{align}
    \hspace{-3pt}\bA_{cl}{=}\barbA{+}\tilbB\hatbK\tilbC,\,
    \bB_{cl}{=}\barbB{+}\tilbB\hatbK\tilbH,\,
    \bC_{cl}{=}\barbC{+}\hatbH\hatbK\tilbC,\,
    \bD_{cl}{=}\hatbH\hatbK\tilbH,\hspace{-5pt}
\end{align}
where the auxiliary system matrices are defined as
\begin{align*}
    \barbA&{=}{\setlength{\arraycolsep}{2pt}\begin{bmatrix}
        \bA_d{+}\bB_d\bH\bC_d & \bzero \\ \bzero & \bzero
    \end{bmatrix}},
    \barbB{=}{\setlength{\arraycolsep}{2pt}\begin{bmatrix}
        \bB_d & \bzero \\ \bzero & \bzero
    \end{bmatrix}},
    \barbC{=}{\setlength{\arraycolsep}{2pt}\begin{bmatrix}
        \bC_d & \bzero \\ \bzero & \bzero
    \end{bmatrix}},
    \hatbK{=}{\setlength{\arraycolsep}{2pt}\begin{bmatrix}
        \hatbA_d & \hatbB_d \\ \hatbC_d & \hatbD_d
    \end{bmatrix}}, \\
    \tilbB&{=}{\setlength{\arraycolsep}{2pt}\begin{bmatrix}
        \bzero & \bB_d\tilbH_{\haty} \\ \bI & \bzero
    \end{bmatrix}},
    \tilbC{=}{\setlength{\arraycolsep}{2pt}\begin{bmatrix}
        \bzero & \bI \\ \tilbH_y\bC_d & \bzero
    \end{bmatrix}},
    \hatbH{=}{\setlength{\arraycolsep}{2pt}\begin{bmatrix}
        \bzero & \bzero \\ \bzero & \tilbH_{\haty}
    \end{bmatrix}},
    \tilbH{=}{\setlength{\arraycolsep}{2pt}\begin{bmatrix}
        \bzero & \bzero \\ \bzero & \tilbH_{y}
    \end{bmatrix}}.
\end{align*}
This attenuation process is formulated as minimizing the $\cH_\infty$-norm bound $\nu{\geq}0$ subject to the existence of $\bY{\succ}0$ satisfying
\begin{align} \label{eq:Hinf_norm}
    \begin{bmatrix}
        \bY\bA_{cl}{+}\bA_{cl}^T\bY & \bY\bB_{cl} & \bC_{cl}^T \\
        \bB_{cl}^T\bY & -\nu\bI & \bD_{cl}^T \\
        \bC_{cl} & \bD_{cl} & -\nu\bI
    \end{bmatrix}\prec0.
\end{align}

\subsection{Sparsity Promoting Controller Synthesis}
Including a penalty term $g(\barbH)$ in the objective function, such as the $\ell_1$ norm, weighted $\ell_1$ norm, sum-of-logs, and the cardinality function of $\barbH$ \cite{lin2013design}, promotes sparsity in the network topology $\barbH$. 
Accordingly, the sparsity-promoting dissipativity augmented controller is obtained by solving
\begin{subequations} \label{eq:sparsity_promotion}
\begin{align}
    \hspace{-7pt}\arg\!\!\min_{\substack{
            \scriptscriptstyle \hatbK,\bF,\hatbF,\barbH
        }}&\quad J(\hatbK)+\gamma g(\barbH) \label{eq:objective}\\[-6pt]
    \sthat\quad&\begin{bmatrix}
        \bY\bA_{cl}{+}\bA_{cl}^T\bY & \bY\bB_{cl} & \bC_{cl}^T \\
        \bB_{cl}^T\bY & -\nu\bI & \bD_{cl}^T \\
        \bC_{cl} & \bD_{cl} & -\nu\bI
    \end{bmatrix}{\prec}0, \label{eq:lyapunov}\\
    &\hspace{-10pt}\bF_i{\in}\set{\bF_i|\text{Constraint for dissipativity holds}},\forall i{\in}\bbN_N \label{eq:diss_plant}\\
    &\hspace{-10pt}{\setlength{\arraycolsep}{2pt}{-}\begin{bmatrix}
        \hatbC_i^T\hatbQ_i\hatbC_i & \hatbC_i^T\hspace{-1pt}\hatbQ_i\hatbC_i \\ \hatbD_i^T\hatbQ_i\hatbC_i & \hatbR_i{+}\hatbD_i^T\hatbQ_i\hatbD_i
    \end{bmatrix}\hspace{-2pt}
    {+}\he\hspace{-1pt}\brackets{\hspace{-1pt}
    \begin{bmatrix}
        \bP_i & \bzero \\ \bzero & {-}\bS_i^T
    \end{bmatrix}
    \hatbK_i}{\preceq}0},\hspace{-5pt} \label{eq:diss_ctrl}\\
    &\bQ{+}\he(\bS\barbH){+}\barbH^T\bR\barbH{\prec}0, \label{eq:diss_ndt} 
\end{align}
\end{subequations}
where $\hatbK_i{=}\begin{bmatrix}
    \hatbA_i & \hatbB_i \\ \hatbC_i & \hatbD_i
\end{bmatrix}$, $\bF_i{=}\set{\bQ_i,\bS_i,\bR_i}$, $\hatbF_i{=}\set{\hatbQ_i,\hatbS_i,\hatbR_i}$, $\bQ{=}\diag(\diag(\bQ_i)_{i\in\bbN_N},\diag(\hatbQ_i)_{i\in\bbN_N})$, $\bS$ and $\bR$ are defined analogously, $\bF{=}\bigcup_{i=1}^N\bF_i$, and $\hatbF{=}\bigcup_{i=1}^N\hatbF_i$.
Simultaneously enforcing \cref{eq:lyapunov} and \cref{eq:diss_ndt} together constitutes the dual-model synthesis approach explained in \cref{subchap:dual_model}.

\subsection{Advantage of \gls{ndt}}
The primary advantage of \gls{ndt} lies in its modular framework, enabling the stabilization of a global network using only open-loop characteristics of agents and their local controllers.
By decoupling agent-level dynamics from the network topology, any heterogeneous agents can be integrated into \cref{eq:sparsity_promotion} once they satisfy dissipativity conditions, which can be formulated into \cref{eq:diss_plant}.
In addition, various dissipativity characterizations can be employed to formulate \cref{eq:diss_plant} as an \gls{lmi}. For instance, \cite[Equation 3.2]{gupta1996robust}, 
\cite[Theorem 3.1]{bridgeman2016conic}, or \cite[Theorem 4]{Strong2024IterativeGain} provide suitable formulations for agents modeled as \gls{lti} systems, 
\gls{lti} systems with input/output/state delays, or input-affine nonlinear systems, respectively.
If dissipativity parameters of nonlinear agents are established a priori, e.g. $(\bQ_i^p,\bS_i^p,\bR_i^p)$, these predefined parameters may be scaled by a design variable $\lambda_i{\geq}0$, yielding $(\bQ_i,\bS_i,\bR_i){=}\lambda_i(\bQ_i^p,\bS_i^p,\bR_i^p)$, and  \cref{eq:diss_plant} can be omitted.


\section{Algorithm} \label{Chap:algorithm}
The difficulty of \cref{eq:sparsity_promotion} lies in its non-convex penalty and the constraints \cref{eq:lyapunov}, \cref{eq:diss_ctrl}, and \cref{eq:diss_ndt}.
This section details an approach to solve \cref{eq:sparsity_promotion} by reformulating \cref{eq:lyapunov}, \cref{eq:diss_ctrl}, and \cref{eq:diss_ndt} using \cref{thm:overbounding} and computing a feasible sparse controller based on \gls{admm}.

\subsection{Convex Overbounding of \texorpdfstring{\gls{bmi}}{BMI} Constraints} \label{subChap:Convex_Overbounding_of_BMI_Constraints}
This section derives convex reformulations of \cref{eq:lyapunov}, \cref{eq:diss_ctrl}, and \cref{eq:diss_ndt} based on \cref{thm:overbounding} using given initial points. 
First, \cref{cor:lmi_Hinf} provides \gls{lmi} implying \cref{eq:lyapunov}.
\begin{corollary} \label{cor:lmi_Hinf}
    Given $\hatbK^0,\tilbB^0,\tilbC^0,\hatbH^0,\tilbK$, and $\bY^0$, if there exist $\dhatbK,\dtilbB,\dtilbC,\dhatbH,\dtilbH,\dbY$, and $\nu{>}0$ such that $\bY^0{+}\dbY{\succ}0$ and
    \begin{align} \label{eq:lmi_Hinf}
        \bracketl{
        \begin{array}{c:c}
            \Phi_{11} & * \\\hdashline
            \Phi_{21} & \Phi_{22}
        \end{array}
        }{\prec}0,
        \end{align}
        where
    \begin{align}
        \Phi_{11}&{=}\begin{bmatrix}
            \he(\cT_0^1(\bY(\barbA{+}\tilbB\hatbK\tilbC))) & * & * \\
            (\cT_0^1(\bY(\barbB{+}\tilbB\hatbK\tilbH)))^T & -\nu\bI & \bzero \\
            \cT_0^1(\barbC{+}\hatbH\hatbK\tilbC) & \cT_0^1(\hatbH\hatbK\tilbH) & -\nu\bI
        \end{bmatrix},\nonumber\\
        \Phi_{21}&{=}{\setlength{\arraycolsep}{3pt}\begin{bmatrix}
            (\bY^0\tilbB^0\dhatbK)^T{+}\dtilbC & \bzero & (\Xi_l\dPi_k)^T \\
            (\bY^0\tilbB^0\dhatbK)^T & \dtilbH & (\hatbH^0\dhatbK)^T \\
            (\dbY^0\dtilbB)^T{+}\dhatbK\tilbC^0{+}\hatbK^0\dtilbC & \dhatbK\tilbH^0{+}\hatbK^0\dtilbH & \bzero \\
            \dhatbK\tilbC^0{+}\hatbK^0\dtilbC & \dhatbK\tilbH^0{+}\hatbK^0\dtilbH & \dhatbH^T \\
            \bL_1^T{+}\dbY & \bL_2^T & \bzero
        \end{bmatrix}}, \nonumber\\
        \Phi_{22}&{=}\begin{bmatrix}
            -2\bI & \bzero & * & * & * \\
            \bzero & -2\bI & * & * & * \\
            \dhatbK & \dhatbK & -2\bI & \bzero & * \\
            \dhatbK & \dhatbK & \bzero & -2\bI & \bzero \\
            (\tilbB^0\dhatbK)^T & (\tilbB^0\dhatbK)^T & (\dtilbB)^T & \bzero & -2\bI
        \end{bmatrix}, \nonumber
    \end{align}
    $\bL_1\hspace{-1pt}{=}\dtilbB\hatbK^0\tilbC^0\hspace{-1pt}{+}\tilbB^0\dhatbK\tilbC^0\hspace{-1pt}{+}\tilbB^0\hatbK^0\dtilbC$, and $\bL_2{=}\dtilbB\hatbK^0\tilbH^0\hspace{-1pt}{+}\tilbB^0\dhatbK\tilbH^0+\tilbB^0\hatbK^0\dtilbH$, then $\hatbK^0\hspace{-2pt}{+}\dhatbK$, $\tilbB^0\hspace{-2pt}{+}\dtilbB$, $\tilbC^0\hspace{-2pt}{+}\dtilbC$, $\hatbH^0\hspace{-2pt}{+}\dhatbH$, $\tilbH^0\hspace{-2pt}{+}\dtilbH$, and $\bY^0\hspace{-2pt}{+}\dbY$ are feasible points of \cref{eq:Hinf_norm}. Moreover, \cref{eq:lmi_Hinf} is always feasible if $\hatbK^0$, $\tilbB^0$, $\tilbC^0$, $\hatbH^0$, $\tilbH^0$, and $\bY^0$ are feasible for \cref{eq:Hinf_norm}.
\end{corollary}
\begin{proof}
    The proof follows by applying the overbounding condition in \cref{eq:overbounding_sebe} of \cref{thm:overbounding} sequentially with $\bG=\bI$. First, use $\hatbK{=}\hatbK^0{+}\dhatbK$ and $\tilbC{=}\tilbC^0{+}\dtilbC$. Next, use $\tilbH{=}\tilbH^0{+}\dtilbH$. Then use $\tilbB{=}\tilbB^0{+}\dtilbB$ and $\hatbH{=}\hatbH_0{+}\dhatbH$. Finally, use $\bY{=}\bY^0{+}\dbY$.
\end{proof}

The \glspl{lmi} in \cref{cor:lmi_kyp,cor:lmi_ndt}, introduced in \cite{locicero2025dissipativity,jang2025communication}, imply \cref{eq:diss_ctrl,eq:diss_ndt}, respectively.

\begin{corollary} \cite{locicero2025dissipativity} \label{cor:lmi_kyp}
    Given $\hatbA^0,\hatbB^0,\hatbC^0,\hatbD^0,\hatbQ^0,\hatbS^0,\hatbR^0$, and $\hatbP^0$, suppose there exist $\dhatbA,\dhatbB,\dhatbC,\dhatbD,\dhatbQ,\dhatbS,\dhatbR$, and $\dbY$ such that $\bY^0{+}\dbY{\succ}0$
    \begin{align} \label{eq:lmi_kyp}
        {\setlength{\arraycolsep}{1pt}
        \hspace{-20pt}\bracketl{
        \begin{array}{cccc}
            \hspace{-2pt}\cT_0^1\hspace{-1pt}\brackets{\hspace{-1pt}
            {-}\hspace{-1pt}\begin{bmatrix}
                \hspace{-1pt}\hatbC^T\hspace{-1pt}\hatbQ\hatbC & \hatbC^T\hspace{-1pt}\hatbQ\hatbC \\ \hspace{-1pt}\hatbD^T\hspace{-1pt}\hatbQ\hatbC & \hatbR\hspace{-1pt}{+}\hatbD^T\hspace{-1pt}\hatbQ\hatbD\hspace{-1pt}
            \end{bmatrix}\hspace{-2pt}
            {+}\he\hspace{-1pt}\brackets{\hspace{-1pt}
            \begin{bmatrix}
                \hspace{-1pt}\bP & \bzero \\ \bzero & {-}\bS^T\hspace{-1pt}
            \end{bmatrix}\hspace{-1pt}
            \hatbK
            \hspace{-1pt}
            }\hspace{-2pt}
            } & \hspace{-5pt}\;\;* & * & *  \\
            \begin{bmatrix}
                \dhatbP & \bzero \\ \bzero & {-}\dhatbS
            \end{bmatrix}
            {+}\dhatbK & \hspace{-7pt}{-}2\bI\hspace{-3pt} & \bzero & \bzero \\
            {-}\frac{1}{2}\hatbQ^0\begin{bmatrix}
                \dhatbC & \dhatbD
            \end{bmatrix} & \bzero & {-}2\bI & * \\
            \begin{bmatrix}
                \dhatbC{-}\dhatbQ\hatbC^0 & \dhatbD{-}\dhatbQ\hatbD^0
            \end{bmatrix} & \bzero & \hspace{-2pt}{-}\frac{1}{2}\dhatbQ & {-}2\bI\hspace{-2pt}
        \end{array}
        }\hspace{-2pt}{\preceq}0}{,}\hspace{-20pt}
    \end{align}
    Then $\hatbA^0\hspace{-2pt}{+}\dhatbA$, $\hatbB^0\hspace{-2pt}{+}\dhatbB$, $\hatbC^0\hspace{-2pt}{+}\dhatbC$, $\hatbD^0\hspace{-2pt}{+}\dhatbD$, $\hatbQ^0\hspace{-2pt}{+}\dhatbQ$, $\hatbS^0\hspace{-2pt}{+}\dhatbS$, $\hatbR^0\hspace{-2pt}{+}\dhatbR$, and $\hatbP^0\hspace{-2pt}{+}\dhatbP$ are feasible points of \cref{eq:diss_ctrl}. Moreover, if $\hatbA^0,\hatbB^0$, $\hatbC^0,\hatbD^0,\hatbQ^0,\hatbS^0,\hatbR^0$, and $\hatbP^0$ are feasible for \cref{eq:diss_ctrl}, \cref{eq:lmi_kyp} is always feasible.
\end{corollary}

\begin{corollary} \cite{jang2025communication} \label{cor:lmi_ndt}
    Given $\bQ^0,\bS^0,\bR^0$, and $\barbH^0$, suppose there exist $\dbQ,\dbS,\dbR$, and $\dbarbH$ such that
    \begin{align} \label{eq:lmi_ndt}
        \begin{bmatrix}
            \cT_0^1(\bQ{+}\he(\bS\barbH){+}\barbH^T\bR\barbH) & * & * & * \\
            \dbS^T{+}\dbarbH & {-}2\bI & \bzero & \bzero \\
            \frac{1}{2}\bR^0\dbarbH{+}\dbH & \bzero & {-}2\bI & * \\
            \dbR\bH^0{+}\dbarbH & \bzero & \frac{1}{2}\dbR & {-}2\bI
        \end{bmatrix}{\prec}0,
    \end{align}
    Then $\bQ^0\hspace{-2pt}{+}\dbQ,\bS^0\hspace{-2pt}{+}\dbS,\bR^0\hspace{-2pt}{+}\dbR$, and $\barbH^0\hspace{-2pt}{+}\dbarbH$ are feasible points of \cref{eq:diss_ndt}. Moreover, if $\bQ^0,\bS^0,\bR^0$, and $\barbH^0$ are feasible for \cref{eq:diss_ndt}, \cref{eq:lmi_ndt} is always feasible.
\end{corollary}

Applying \cref{cor:lmi_Hinf,cor:lmi_kyp,cor:lmi_ndt} yields the convex problem
\begin{subequations} \label{eq:lmi_control_synthesis}
\begin{align}
    \hspace{-7pt}\arg\!\!\!\!\!\min_{\substack{
            \scriptscriptstyle \dhatbK,\dbF,\dhatbF,\dbarbH
        }}&J(\hatbK), \label{eq:lmi_objective}\\[-6pt]
    \sthat\quad&\bN(\dhatbK,\dbF,\dhatbF,\dbarbH){\preceq}0, \label{eq:lmi_constraint_agent}
\end{align}
\end{subequations}
to construct the centralized controller for given initial feasible points, 
where $\bN(\hspace{-1pt}\dhatbK,\dbF,\dhatbF,\dbarbH\hspace{-1pt})$ composes: (\labelcref{eq:lmi_Hinf}) to ensure that $\cH_\infty$ norm of the network; agent dissipativity requirements encoded as \glspl{lmi}; \cref{eq:lmi_kyp} to impose controller dissipativity; and \cref{eq:lmi_ndt} to ensure that the network is stable via \gls{ndt}.

\subsection{Initialization} \label{subChap:Initialization}
Solving \cref{eq:lmi_control_synthesis} requires a set of initial feasible points, $\{\hatbA_i^0{,}\hatbB_i^0{,}\hatbC_i^0{,}\hatbD_i^0\}$ for $i{\in}\bbN_N$ and $\{\bQ^0{,}\bS^0{,}\bR^0{,}\barbH^0\}$.  If all agents are open-loop stable, the technique in \cite[Section 6]{locicero2025dissipativity} provides this. 
Otherwise, local controllers $\scrC_i$ can be designed via standard synthesis procedures, such as PID, \gls{lqg}, or $\cH_\infty$ design, while assuming a decentralized structure by fixing $\tilbH_y{=}\bI$ and $\tilbH_{\haty}{=}\bI$.
If the initial design fails to satisfy the dissipativity/stability constraints, the feasibility test is repeated with increased controller gains until a valid feasible point is achieved. 
The iterative relaxation approach of \cite{warner2017iterative} can also be used.
After feasible local controllers are obtained, solving \cref{eq:lmi_control_synthesis} yields an initial centralized feasible controller with dense initial interconnection matrices $\tilbH_y^0$ and $\tilbH_{\haty}^0$.

\subsection{Sparsity Promotion} \label{subChap:Sparsity_Promotion}
Once the initial centralized feasible points are obtained, the sparse topology can be determined by 
\begin{subequations} \label{eq:lmi_sparse_control_synthesis}
\begin{align}
    \hspace{-7pt}\arg\!\!\!\min_{\substack{
            \scriptscriptstyle \dhatbK,\dbF,\dhatbF,\dbarbH
        }}&J(\dhatbK)+\gamma g(\barbH^0+\dbarbH), \label{eq:sparse_objective}\\[-6pt]
    \sthat\quad&\bN(\dhatbK,\dbF,\dhatbF,\dbarbH){\preceq}0\label{eq:sparse_constraint_agent}
\end{align}
\end{subequations}
where the penalty function $g(\barbH^0{+}\dbarbH)$ promotes sparsity in the global controller. 
Following \cite[Section IV C.]{jang2025communication}, we consider the weighted $\ell_1$, $g_1(\barbH)$, and cardinality penalty, $g_0(\barbH)$,
\begin{align}
    \hspace{-5pt}g_1(\barbH^0\hspace{-3pt}{+}\dbarbH)&{=}\hspace{-10pt}\sum_{i,j\in\bbN_{2N}}\hspace{-9pt}\min\hspace{-1.5pt}\big\{{\|}(\barbH^0\hspace{-3pt}{+}\dbarbH)_{ij}{\|}_F^{{-}1},\epsilon_l^{{-}1}\big\}{\|}(\barbH^0\hspace{-3pt}{+}\dbarbH)_{ij}{\|}_F,\hspace{-5pt}\label{eq:weighted_ell1} \\
    \hspace{-5pt}g_0(\barbH^0\hspace{-3pt}{+}\dbarbH)&{=}\hspace{-10pt}\sum_{i,j\in\bbN_{2N}}\hspace{-9pt}\card\big({\|}(\barbH^0\hspace{-3pt}{+}\dbarbH)_{ij}{\|}_F\big). \label{eq:cardinality_penalty}
\end{align}
The problem is convex with \cref{eq:weighted_ell1}, but not with \cref{eq:cardinality_penalty}.
The \gls{admm} iteration from \cite[Section IV C.]{jang2025communication} can be employed with \cref{eq:cardinality_penalty} as
\begin{subequations} 
    \begin{align}
        \hspace{-5pt}\dbarbH^{r+1}&{=}\arg\!\!\!\!\!\!\min_{\substack{
            \scriptscriptstyle \dhatbK,\dbF,\dhatbF,\dbarbH
        }}J(\dhatbK){+}\frac{\rho}{2}\|\barbH^0{+}\dbarbH{-}\bZ^r{+}\bm\Lambda^r\|_F^2 \label{eq:H_update}\\[-4pt]
            &\qquad\quad\sthat\;\;\;\bN(\dhatbK,\dbF,\dhatbF,\dbarbH),\nonumber\\
        \bZ^{r+1}&{=}\arg\min_\bZ\gamma g(\bZ){+}\frac{\rho}{2}\|\barbH^0{+}\dbarbH^{r+1}{-}\bZ{+}\bm\Lambda^r\|_F^2{,}  \label{eq:Z_update}\\
        \bm\Lambda^{r+1}&{=}\bm\Lambda^r{+}(\barbH^0+\dbarbH^{r+1}{-}\bZ^{r+1}), \label{eq:Lambda_update}
    \end{align}
\end{subequations}
where $\bZ$ is the clone of $\barbH^0{+}\dbarbH$, $\bm\Lambda$ is the dual variable, $r$ is the iteration index of \gls{admm}, and $\rho{>}0$ is the augmented Lagrangian parameter.
The block component of the unique solution to \cref{eq:Z_update}, $(\bZ^{r+1})_{ij}$, is $(\bV)_{ij}$ if $\|(\bV)_{ij}\|_F{>}\sqrt{2\gamma/\rho}$ and $\bzero$ otherwise,
where $\bV{=}\barbH^0{+}\dbarbH^{r+1}{+}\bm\Lambda^r$ \cite{lin2013design}. The stopping criteria of \gls{admm} are $r_p{=}\frac{\|\barbH^0{+}\dbarbH^{r}{-}\bZ^{r}\|_F}{\|\bZ^{r}\|_F}{\leq}\epsilon_p$ and $r_d{=}\frac{\|\bZ^{r}{-}\bZ^{r-1}\|_F}{\|\bZ^{r}\|_F}{\leq}\epsilon_d$.

After solving \cref{eq:lmi_sparse_control_synthesis}, the local controllers and networks are updated to  $\hatbA_i^1{=}\hatbA_i^0\hspace{-2pt}{+}\dhatbA_i^\star$, $\hatbB_i^1{=}\hatbB_i^0\hspace{-2pt}{+}\dhatbB_i^\star$, $\hatbC_i^1{=}\hatbC_i^0\hspace{-2pt}{+}\dhatbC_i^\star$, $\hatbD_i^1{=}\hatbD_i^0\hspace{-2pt}{+}\dhatbD_i^\star$ for all $i{\in}\bbN_N$, and $\barbH^1{=}\barbH^0\hspace{-2pt}{+}\dbarbH^\star$, where $\dhatbA_i^\star$, $\dhatbB_i^\star$, $\dhatbC_i^\star$, $\dhatbD_i^\star$  for all $i{\in}\bbN_N$, and $\dbarbH^\star$ are the optimal perturbation obtained from the solution.  
The matrices $\bQ^1$, $\bS^1$, and $\bR^1$ are updated accordingly.
The subspace $\cH$ is then defined as the set of block matrices that share the same sparsity pattern as $\barbH^1$.

\subsection{Structured \texorpdfstring{\gls{ico}}{ico}} \label{subChap:Structured_Iterative_Convex_Overbounding}
Once the sparse structure $\cH$ is identified, we consider 
\begin{subequations} \label{eq:main_problem_structured}
    \begin{align}
    \hspace{-7pt}\arg\!\!\!\!\!\min_{\substack{
            \scriptscriptstyle \dhatbK,\dbF,\dhatbF,\dbarbH
        }}&\sum_{i=1}^NJ_i(\dhatbK_i), \label{eq:structured_objective}\\[-4pt]
    \sthat\quad&\bN(\dhatbK,\dbF,\dhatbF,\dbarbH){\preceq}0,\;\;\barbH^k{+}\dbarbH{\in}\cH, \label{eq:structured_constraint}
    \end{align}
\end{subequations}
and apply \gls{ico} to \cref{eq:main_problem_structured} to determine the optimal controller parameters $\hatbK^\star$, using the feasible points $\{\hatbA_i^1{,}\hatbB_i^1{,}\hatbC_i^1{,}\hatbD_i^1\}$  for all $i{\in}\bbN_N$, and $\{\bQ^1{,}\bS^1{,}\bR^1{,}\barbH^1\}$ obtained from the process in \cref{subChap:Sparsity_Promotion}. 
\gls{ico} is applied to \cref{eq:main_problem_structured} instead of \cref{eq:lmi_sparse_control_synthesis} since the cardinality penalty's discontinuity at zero invalidates the convergence guarantee of \gls{ico} \cite{locicero2022sparsity}.
The overall procedure of \cref{Chap:algorithm} is summarized in \cref{alg:sparsity_promotion}.

\begin{algorithm}[tbp]
    \caption{Sparsity-Promoting Controller Synthesis}\label{alg:sparsity_promotion}
    \begin{algorithmic}[1]
        \Require $g{,}\hatbK_i^0$ for $i\in\bbN_N,\bQ^0,\bS^0,\bR^0,\barbH^0,\epsilon_p,\epsilon_d,\epsilon_l,\epsilon$
        \Ensure $\hatbK_i$ for $i\in\bbN_N$, $\barbH$
        \If {$g(\barbH^0+\dbarbH)$ follows \cref{eq:weighted_ell1}}
            find $\dbarbH^\star$ by solving \cref{eq:lmi_sparse_control_synthesis}
        \ElsIf {$g(\barbH^0+\dbarbH)$ follows \cref{eq:cardinality_penalty}}
            \While {$r_p>\epsilon_p$ or $r_d>\epsilon_d$}
                \State Find $\dbarbH^{k+1}$ by solving \cref{eq:H_update}
                \State Find $\bZ^{k+1}$ by solving \cref{eq:Z_update}
                \State Find $\bm\Lambda^{k+1}$ by solving \cref{eq:Lambda_update}
            \EndWhile
        \EndIf
        \State Set $\barbH^1{=}\barbH^0{+}\dbarbH^\star$, and update feasible points, $\hatbK_i^1$ for $i\in\bbN_N$, $\bQ^1$ $\bS^1$, and $\bR^1$  analogous to $\barbH^1$.
        \State Define the structured subspace $\cH$
        \While {$\frac{J_i(\dhatbK_i^k)-J_i(\dhatbK_i^{k+1})}{J(\dhatbK_i^{k+1})}\nleq\epsilon$ for all $i\in\bbN_N$}
            \State Solve \cref{eq:main_problem_structured} using $\hatbK_i^k$ for $i{\in}\bbN_N$, $\bQ^k,\bS^k,\bR^k$, and $\barbH^k$
            \State Set $\barbH^{k+1}{=}\barbH^k{+}\dbarbH^\star$, and other feasible points, $\hatbK_i^{k+1}$, for $i{\in}\bbN_N,\bY^{k+1},\bQ^{k+1},\bS^{k+1}$, and $\bR^{k+1}$ analogous to $\barbH^1$
        \EndWhile
    \end{algorithmic}
\end{algorithm}

\subsection{Convergence of \cref{alg:sparsity_promotion}}
\cref{alg:sparsity_promotion} combines two stages, sparsity promotion and structured \gls{ico}.
The \gls{ico} stage is recursively feasible and converges to a local optimum, as each iteration is initialized with the previous feasible point and guarantees a non-increasing cost.
In contrast, the convergence of the sparsity-promotion stage depends on the penalty function:
$g_1$ ensures convergence because the problem is formulated as a \gls{sdp}, but $g_0$ is nonconvex and therefore lacks convergence guarantees.
However, as noted in \cite{lin2013design}, \gls{admm} converges regardless with a sufficiently large $\rho$.

\section{Numerical Example}
Sparse controllers are synthesized using \cref{alg:sparsity_promotion} for a networked system with polytopic uncertainty. 
We consider $N{=}10$ agents randomly distributed over a $2{\times}2$ grid. In a modified version of \cite{motee2008optimal}, the individual agents' dynamics deviate randomly from their nominal models,
\begin{align*}
    \dbx_i{=}\bA_{i}^n\bx_i{+}\begin{bmatrix}
        0 \\ 1
    \end{bmatrix}\Big(\be_i{+}\sum_{j\neq i}e^{-\alpha(i-j)}\by_j\Big),\quad
    \by_i{=}\begin{bmatrix}
        1 & 1
    \end{bmatrix}\bx_i,
\end{align*}
where $\bA_{i}^n=\begin{bmatrix}
    1 & 1 \\ 1 & 2
\end{bmatrix}$ for $i{\in}\bbN_5$ and $\begin{bmatrix}
    -2 & 1 \\ 1 & -3
\end{bmatrix}$ otherwise, $\alpha{=}0.1823$, $\bx_i{\in}\bbR^2$, and $\be_i{\in}\bbR$.
The first 5 agents are nominally unstable, while the remaining 5 are stable.
From local dynamics, $\bH$ has blocks $(\bH)_{ii}{=}0$ and $(\bH)_{ij}{=}e^{-\alpha(i-j)}$ for $i{\neq}j$.

The actual dynamics $\bA_{i}$ of each agent are uniformly distributed within $\pm4\%$ of their nominal values $\bA_{i}^n$, so each agent's dynamics can be modeled as an \gls{lti} system with polytopic uncertainty, represented by a polyhedron with 16 vertices, where each vertex corresponds to the maximum or minimum of one parameter in $\bA_{i}$.
To enforce agent dissipativity, \cite[Lemma 2]{jang2025communication} is used in \cref{eq:diss_plant}, with a small adjustment to generalize from $\by{=}\bx$ in \cite{jang2025communication} to $\by{=}\bC\bx$ here.

After computing the local \gls{lqg} controller with $\bQ_n{=}100\bI_2$ and $\bR_n{=}1$, an initial local feasible controller is found by scaling the transfer function of the \gls{lqg} controller by $10^{-3}$ and adding feedforward gain ${-}10$ for all $10$ agents.
The parameters are $\epsilon{=}\epsilon_p{=}\epsilon_d{=}\epsilon_l{=}10^{-3}$, and $\rho{=}1000$.
The weighting factor $\gamma$ is varied over 20 logarithmically spaced points on the inverval $[2{\times}10^{-5}{,}1.5]$ for the weighted $\ell_1$ penalty, and $[1{\times}10^{-6}{,}5]$ for the cardinality penalty.

\begin{figure}
\centering
    \subfloat[Different sparsity levels]{
    \resizebox{0.225\textwidth}{!}{
%
%
\definecolor{mycolor1}{rgb}{1.00000,0.00000,1.00000}%
\begin{tikzpicture}

\begin{axis}[%
width=4.521in,
height=3.566in,
at={(0.758in,0.481in)},
scale only axis,
xmin=20,
xmax=200,
xlabel style={font=\color{white!15!black}},
xlabel={Nonzero blocks in $\tilbH_y$ and $\tilbH_{\haty}$},
ymin=16.5,
ymax=23,
ylabel style={font=\color{white!15!black}},
ylabel={$J(\hatbK)$},
axis background/.style={fill=white},
legend style={at={(0.97,0.5)}, anchor=east, legend cell align=left, align=left, draw=white!15!black},
title style={font=\Huge},xlabel style={font={\LARGE}},ylabel style={font=\LARGE},legend style={font=\LARGE},
]
\addplot [color=blue, line width=1.0pt, mark size=2.5pt, mark=square*, mark options={solid, fill=blue, blue}]
  table[row sep=crcr]{%
20	22.6608126486484\\
24	19.4876976312872\\
36	17.7093921459345\\
40	17.787661185665\\
46	17.424365980529\\
80	17.1060119824429\\
82	17.1119932289834\\
84	17.1005690683585\\
86	17.1027749441875\\
88	17.093863848263\\
94	17.1019342194381\\
98	17.1173739082966\\
100	17.1312640021908\\
118	17.0975522758024\\
120	17.0862388100349\\
154	17.0961131974883\\
192	17.0868210440078\\
196	17.0759428427794\\
198	17.0802880868547\\
200	17.0878478528421\\
};
\addlegendentry{Weighted $\ell_1$}

\addplot [color=green, line width=1.0pt, mark size=3.5pt, mark=*, mark options={solid, fill=green, green}]
  table[row sep=crcr]{%
20	22.6608126486484\\
22	19.4980839320791\\
36	17.7112728723129\\
76	17.1798575613299\\
102	17.1110782264598\\
108	17.1563962389022\\
120	17.1195441028011\\
124	17.0874365115373\\
126	17.0959312053772\\
130	17.0924934902985\\
134	17.0870707282261\\
136	17.0930520090212\\
146	17.1001230980794\\
166	17.0865217093573\\
182	17.0961518663799\\
188	17.0824375573032\\
192	17.0919574704452\\
194	17.0882450626288\\
196	17.0761964223348\\
200	17.0878478528421\\
};
\addlegendentry{Cardinality}

\addplot [color=mycolor1, dotted, line width=3.0pt]
  table[row sep=crcr]{%
20	17.0878478528421\\
200	17.0878478528421\\
};
\addlegendentry{Centralized}

\addplot [color=black, dotted, line width=3.0pt]
  table[row sep=crcr]{%
20	22.6608126486484\\
200	22.6608126486484\\
};
\addlegendentry{Decentralized \cite{locicero2025dissipativity}}

\addplot [color=red, dotted, line width=3.0pt]
  table[row sep=crcr]{%
20	17.3664960926324\\
200	17.3664960926324\\
};
\addlegendentry{5\% threshold}

\end{axis}
\end{tikzpicture}%
} \label{fig:Hinf_performance}
    }
    \subfloat[Different uncertain agents]{
    \resizebox{0.225\textwidth}{!}{\input{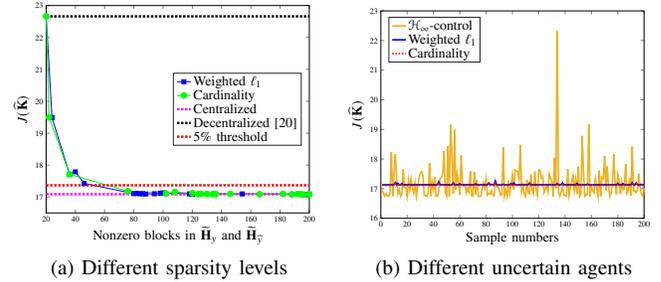}} \label{fig:Hinf_uncertain}
    }
    \caption{$\cH_\infty$-norm of resulting sparse controllers; In \cref{fig:Hinf_uncertain}, the number of nonzero blocks in $\tilbH_y$ and $\tilbH_{\haty}$ is $100$ and $102$ for weighted $\ell_1$ norm and cardinality, respectively.}  \label{fig:performance_analysis}
    \vspace*{-1\baselineskip} 
\end{figure}

\begin{table}
\begin{center}
\caption{Best and worst $\cH_\infty$-norm of \cref{fig:Hinf_uncertain}}\label{tb:Hinf}
\begin{tabular}{c:ccc}
\hline\hline
 & $\cH_\infty$-control & Weighted $\ell_1$ & Cardinality \\ \hline
Best & 16.7350 & 17.1313 & 17.1111 \\
Worst & 22.3155 & 17.2547 & 17.2325 \\ \hline\hline
\end{tabular}
\end{center}
\vspace*{-2.5\baselineskip} 
\end{table}

\cref{fig:Hinf_performance} shows the $\cH_\infty$ norm of the resulting controllers for different sparsity levels.
Since the parameter $\rho$ is sufficiently large, the sparsity-promotion method using cardinality successfully converges to a sparse controller structure.
The $\cH_\infty$ norm obtained using a decentralized controller network, which is equivalent to the approach in \cite{locicero2025dissipativity}, is $22.66$, whereas the centralized controller achieves $17.09$. 
Starting from the fully decentralized setup, both methods rapidly reduce the $\cH_\infty$ norm and reach a value within 5\% of the gap between decentralized and centralized case, which is $17.37$.
Before reaching the $5\%$ threshold, the cardinality-based method achieves better performance than the weighted $\ell_1$-norm for the same sparsity levels.
After reaching the $5\%$ threshold, the two methods exhibit similar performance.

The $\cH_\infty$ norms of resulting closed-loop networks are evaluated over 200 randomly generated true dynamics of each agent. For the comparison, a standard $\cH_\infty$ control approach is used to compute the optimal $\cH_\infty$ controller for the nominal agent dynamics.
As shown in \cref{fig:Hinf_uncertain}, the $\cH_\infty$ norm obtained using the $\cH_\infty$ controller exhibits significant fluctuations, whereas the norms obtained using the proposed approaches remain nearly constant while requiring half as many communication links.
This behavior arises because the dissipativity constraint accounts for all admissible realizations of the true agent dynamics.
Consequently, \cref{tb:Hinf} shows that the worst-case $\cH_\infty$ norm achieved by the proposed approach is significantly smaller than that obtained by the standard optimal $\cH_\infty$ controller.

\section{Conclusion}
This paper presented the approaches of synthesizing dissipativity-based dynamic output feedback controllers and a sparse communication network simultaneously using \gls{ndt}.
This extends \cite{jang2025communication} from state to output feedback. 
We first construct a condition to ensure the global controller is well-posed. 
Under this well-posed condition, the optimization problem with \gls{ndt} and a sparsity penalty was solved by mixing \gls{admm} and \gls{ico}.
A numerical example showed that the cardinality penalty performs slightly better than weighted $\ell_1$, which follows the conclusion in \cite{jang2025communication}.










\bibliographystyle{IEEEtran}
\bibliography{MyBib}{}
\end{document}